\documentclass[letterpaper, 10 pt, conference]{ieeeconf}  

\IEEEoverridecommandlockouts                              
\usepackage{epsfig} 
\usepackage{mathptmx} 
\usepackage{times} 
\usepackage{amsmath} 
\usepackage{amssymb}  
\usepackage{bm}
\usepackage{comment}
\usepackage{booktabs}
\usepackage{array}
\usepackage[T1]{fontenc}
\usepackage{pgf}
\usepackage{tikz}
\usepackage{tikz-3dplot}
\usepackage{graphicx} 

\usepackage{amsthm}
\newtheorem{theorem}{Theorem}[section]
\newtheorem{corollary}{Corollary}[section]
\newtheorem{remark}{Remark}[section]
\newtheorem{lemma}{Lemma}[section]
\newtheorem{proposition}{Proposition}[section]
\theoremstyle{definition}
\newtheorem{definition}{Definition}[section]

\usepackage[
  colorlinks=true,
  linkcolor=blue!70!black,
  citecolor=green!50!black,
  urlcolor=blue!80!black,
  bookmarks=true,
  bookmarksnumbered=true,
]{hyperref}
\usepackage[nameinlink]{cleveref}

\DeclareMathOperator{\vol}{\mathrm{vol}}

\title{\LARGE \bf
    BLISS: Global Blind Identification of Linear Systems \\ with Sparse Inputs
}

\author{Kyle Poe, Uday Kiran Reddy Tadipatri, Benjamin D. Haeffele, and Ren\'e Vidal
 \thanks{This work was supported by ONR MURI grant 503405-78051 and University of Pennsylvania startup funds.}
 \thanks{The authors are with the Center for Innovation in Data Engineering and Science, University of Pennsylvania. 3333 Chestnut St., Philadelphia PA 19104, USA.
         {\tt\small kpoe@sas.upenn.edu}}%
}

\begin{document}

\maketitle
\thispagestyle{empty}
\pagestyle{empty}

\begin{abstract}

    Linear system identification and sparse dictionary learning can both be seen as structured matrix factorization problems. However, these two problems have historically been studied in isolation by the systems theory and machine learning communities. Although linear system identification enjoys a mature theory when inputs are known, blind linear system identification remains poorly understood beyond restrictive settings. In contrast, complete sparse dictionary learning has recently benefited from strong global identifiability results and scalable nonconvex algorithms. In this work, we bridge these two areas by showing that under a sparse input assumption, fully observed blind system identification becomes a generalization of complete dictionary learning. This connection allows us to develop global identifiability guarantees for blind system identification, by leveraging techniques from the complete dictionary learning literature. We further show empirically that a principled application of the alternating direction method of multipliers can globally recover the ground-truth system from a single trajectory, provided sufficient samples and input sparsity.
    
\end{abstract}

\section{Introduction}
System identification--the problem of learning the parameters of a dynamical system from observed input-output data--remains a central problem in systems and control. Given state observations $(\bm{x}^{(0)}, \cdots, \bm{x}^{(T)})$ and inputs $(\bm{u}^{(0)}, \cdots, \bm{u}^{(T-1)})$ of a ground truth fully-observed linear time invariant (LTI) system
\begin{equation}\label{eq:LTI_eqs}
    \bm{x}^{(k+1)} = \bm{A}_\natural\bm{x}^{(k)} + \bm{B}_\natural\bm{u}^{(k)},\quad \bm{A}_\natural \in \mathbb{R}^{n \times n}, \quad \bm{B}_\natural \in \mathbb{R}^{n \times m},
\end{equation}
we may define the matrices $\bm{X}_+ = [\bm{x}^{(1)}, \cdots, \bm{x}^{(T)}] \in \mathbb{R}^{n \times T}$, $\bm{X} = [\bm{x}^{(0)}, \cdots, \bm{x}^{(T-1)}]^{\top}\in\mathbb{R}^{T \times n}$, and $\bm{U} = [\bm{u}^{(0)}, \cdots, \bm{u}^{(T-1)}]^\top \in \mathbb{R}^{T \times m}$, so that the identification problem can be written as a linear system where one aims to \textit{determine} $\bm{A}_\star, \bm{B}_\star$ such that 
\begin{equation}   
    \bm{X}_+ = \bm{A}_\star\bm{X}^\top + \bm{B}_\star\bm{U}^\top.
\end{equation}

When $\bm{U}$ is known, this problem is well-understood--the necessary and sufficient conditions for unique recovery of $(\bm{A}_\star, \bm{B}_\star) = (\bm{A}_\natural, \bm{B}_\natural)$ is that the matrix $[\bm{X}, \bm{U}]^\top \in \mathbb{R}^{(n +m) \times T}$ have full row rank \cite{willems-et-al-scl05}. Even in the presence of noise and when only indirect state measurements $\bm{y}^{(k)} = \bm{C}\bm{x}^{(k)}$ are available, assuming that $\bm{A}_\natural, \bm{B}_\natural, \bm{C}_\natural$ are a minimal realization, recent work has established sharp non-asymptotic guarantees for recovering the Markov parameters via a simple least squares estimator \cite{finite_sarkar_2021, nonasymptotic_oymak_2018}, by which these matrices can be uniquely recovered up to a similarity transformation. Other recent advances include enhanced guarantees for convex programming under sparsity assumptions on the system matrices \cite{learning_fattahi_2019, sample_fattahi_2021}, and globally certifiable methods for directly learning a minimal-order realization via nonconvex optimization \cite{tadipatri2025nonconvex}.

However, in many cases of practical interest, the inputs $\bm{U}$ may not be known, but possess additional structure such as being \textit{sparse}.  Examples in the natural sciences range among diffusion processes \cite{monge2020sparse}, activation sparsity in biological neural networks \cite{vinje2000sparse}, and seismic signal processing \cite{taylor1979deconvolution}; while examples in the information sciences include fault detection \cite{gajjar2018real}, actuator scheduling \cite{siami2020deterministic}, and neural network pruning \cite{hoefler2021sparsity}, to name only a few. In such cases when the system of interest is known, a variety of recent work has established improved algorithms and conditions for sparse input recovery \cite{chakraborty2023joint, poe2023necessary, poe2024invertibility}, delivering superior guarantees and practical performance.

In this work, we aim to develop global recovery guarantees for such \textit{blind} system identification problems, where in addition to recovering $\bm{A}_\natural, \bm{B}_\natural$, we also aim to recover an unknown, sparse $\bm{U}_\natural$. In this case, the problem becomes: given $\bm{X}_+, \bm{X}$, recover $\bm{A}_\star, \bm{B}_\star, \bm{U}_\star$ such that
\begin{equation}
    \text{$\bm{X}_+ = \bm{A}_\star\bm{X}^\top + \bm{B}_\star\bm{U}_\star^\top$ and $\bm{U}_\star$ is sparse.}
\end{equation}
If $\bm{A}_\natural$ was known, this problem would be equivalent to factorizing an innovation matrix $\bm{R} := \bm{X}_+ - \bm{A}_\natural\bm{X}^\top$ as $\bm{B}_\star\bm{U}_\star^\top$, which is exactly \textit{sparse dictionary learning}, a classical problem in signal processing \cite{tovsic2011dictionary}. Despite the intrinsic nonconvexity of this problem, global identifiability up to scaling and permutation of the columns of the $(\bm{B}_\natural, \bm{U}_\natural)$ matrices in this setting has been established for several different optimization formulations \cite{zhai2020complete, sun2015complete, hu2023global}. Of particular note is the matrix volume optimization paradigm put forward by Hu et al. \cite{hu2023global}, which is unique among existing approaches in that it provides interpretable conditions for \textit{exact} and simultaneous identifiability of the entire dictionary. 
However, in the blind system identification context, $\bm{R}$ can vary due to the unknown $\bm{A}$, so previous results do not apply.

Prior work on joint recovery of $(\bm{A}_\natural, \bm{B}_\natural, \bm{U}_\natural)$ under sparse input assumptions is itself quite sparse. Relevant existing work includes finite impulse response recovery using higher-order-statistics \cite{algorithms_demoor_1999}, single-input single-output deconvolution \cite{blind_ahmed_2014}, and a single paper on recovery of a linear time varying system from state observations in the presence of unknown sparse inputs, when $\bm{B}_\natural = \bm{I}$ \cite{dobbe2019blind}. 
To the best of our knowledge, this work establishes the first global identifiability results for general, unstructured linear systems under unknown sparse inputs.
Our contributions are as follows:

\begin{enumerate}
    \item We formulate the blind system identification problem in the fully observed setting as a constrained nonconvex volume maximization problem.
    \item We establish global identifiability of the ground truth system for this problem under the \textit{persistent scattering} condition, strengthening conditions of both dictionary learning and linear system identification.
    \item We introduce a generalized alternating direction method of multipliers (ADMM) for solving this problem, leveraging efficient closed form proximal updates with bounded iterates.
    \item Our theory is validated by experiments, demonstrating a clear phase transition between perfect, global recovery and failure, strictly improving on average with sparser inputs and larger $T$.
\end{enumerate}

\subsection{Notation}
For an integer $K>0$, we denote $[K]:= \{0, 1, \cdots, K-1\}$. We denote the $\ell_p$ norm by $\|\cdot\|_p$ for $p \in [1, \infty]$, and the closed unit $\ell_p$ ball by $\mathcal{B}_p:= \{\bm{x}: \|\bm{x}\|_p \le 1\}$, where the ambient dimension is inferred from context. We use parenthesized superscripts to index collections of matrices/vectors, and subscripts to denote their columns/entries, such that $\bm{A}^{(k)}\bm{x}^{(k)} = \sum_{i} \bm{A}^{(k)}_i \bm{x}^{(k)}_i$. For a matrix $\bm{B} \in \mathbb{R}^{n \times m}$ with $n \ge m$, we write $\vol(\bm{B}) := \sqrt{\det(\bm{B}^\top \bm{B})}$, and for a square matrix $\bm{\Phi}$ denote $|\bm{\Phi}| := |\det(\bm{\Phi})|$. $\bm{A}^\dagger$ denotes the Moore-Penrose pseudoinverse of $\bm{A}$. We denote the element-wise (Hadamard) product between vectors with $\odot$. For convex analysis, we follow the conventions of Rockafellar \cite{rockafellar1997convex}.

\section{Background}
\subsection{Fully Observed Linear System Identification}\label{sect:sysid}
Linear system identification may be understood as the problem of recovering a minimal system realization from a sequence of state measurements $(\bm{y}^{(k)})_{k \in [T]}$ and inputs $(\bm{u}^{(k)})_{k \in [T]}$. When such measurements uniquely constrain the states, we refer to the system as \textit{fully observed}, and the problem reduces to the recovery of $\bm{A}_\natural, \bm{B}_\natural$ from $(\bm{x}^{(k)}, \bm{u}^{(k)})_{k \in [T]}$, where $\bm{x}^{(k+1)} = \bm{A}_\natural \bm{x}^{(k)} + \bm{B}_\natural \bm{u}^{(k)}$. In this case, the necessary and sufficient condition for unique recovery of $\bm{A}_\natural, \bm{B}_\natural$ is that the matrix $[\bm{X}, \bm{U}] \in \mathbb{R}^{T \times (n+m)}$ be full column rank. This condition is sometimes referred to as $(\bm{x}, \bm{u})$ being \textit{persistently exciting}, in light of the following characterization:
\begin{corollary}[Willems \cite{willems-et-al-scl05}]
    Suppose the system $(\bm{A}_\natural, \bm{B}_\natural)$ is controllable, and that $\bm{u}$ is \textit{persistently exciting} of order $n+1$:
    \begin{equation}
        \text{rank}\begin{bmatrix}
            \bm{u}^{(0)} & \bm{u}^{(1)} & \cdots & \bm{u}^{(T-n-1)} \\
             \vdots & \vdots & & \vdots\\
             \bm{u}^{(n)} & \bm{u}^{(n+1)} & \cdots & \bm{u}^{(T-1)}
        \end{bmatrix} = (n+1)m
    \end{equation}
    Then $\text{rank}\,[\bm{X}, \bm{U}] = n + m$.
    
\end{corollary}

\subsection{Identifiability of Complete Dictionary Learning}
    Sparse dictionary learning is the problem of approximately factorizing a given matrix $\bm{R}$ as $\bm{R} \approx \bm{\Psi} \bm{Z}^\top$ such that $\bm{Z}$ is sparse. When the matrix $\bm{\Psi}$ is sought to be full column rank, this is known as \textit{complete} dictionary learning. In the exact setting, when there exists a true sparse $\bm{Z}_\natural$ and $\bm{\Psi}_\natural$ satisfying $\bm{R} = \bm{\Psi}_\natural \bm{Z}_\natural^\top$, this problem may be solved through an optimization problem of the form
    \begin{equation}
        (\bm{\Psi}_\star, \bm{Z}_\star) \in \arg\min_{\bm{\Psi}, \bm{Z}} f_{\bm{R}}(\bm{\Psi}, \bm{Z}) \text{ s.t. } \bm{R} = \bm{\Psi}\bm{Z}^\top,
    \end{equation}
    where $f_{\bm{R}}$ is an objective function, possibly depending on $\bm{R}$, which promotes sparsity in $\bm{Z}$. Recent work has studied geometric conditions on the sparse representation matrix $\bm{Z}_\natural$ such that for a suitable $f_{\bm{R}}$, any such $(\bm{\Psi}_\star, \bm{Z}_\star)$ is equivalent to $(\bm{\Psi}_\natural, \bm{Z}_\natural)$ in the following sense:

    \begin{definition}[SP-Equivalent]
        We will write $(\bm{\Psi}, \bm{Z}) \cong (\bm{\Psi}', \bm{Z}')$ and say that the pairs are \textit{scale-permutation equivalent} when there exists a nonsingular diagonal matrix $\bm{D}$ and permutation matrix $\bm{\Pi}$ such that
        \begin{equation}
            (\bm{\Psi}', \bm{Z}') = (\bm{\Psi} \bm{D} \bm{\Pi}, \bm{Z} \bm{D}^{-1} \bm{\Pi})         
        \end{equation}
    \end{definition}
    
    Specifically, define the following notion for a set $\mathcal{C}$:
    
    \begin{definition}[Sufficient Scattering]
        We say that a subset $\mathcal{C}$ of the unit hypercube $\mathcal{B}_\infty \subset \mathbb{R}^{m}$ is \textit{sufficiently scattered} if
        \begin{equation}
            \mathcal{B}_2 \subseteq \mathcal{C}, \quad \partial\mathcal{B}_2 \cap \partial \mathcal{C} = \{\pm \bm{e}^{(i)}\}_{i \in [m]}.
         \end{equation}
    \end{definition}
    This condition geometrically captures the notion that the set $\mathcal{C}$ is in some sense \textit{spread out} enough such as to contain the unit sphere, and intersects it only at the unit vectors required by containment in the unit hypercube.
    Leveraging this characterization, the authors prove identifiability of the ground truth dictionary and sparse representation matrix under a constrained volume minimization problem:

    \begin{theorem}[Hu and Huang \cite{hu2023global}]
        Let $\bm{\Lambda}$ be a diagonal matrix such that $\bm{\Lambda}_{ii} = \|(\bm{Z}_\natural)_i\|_1$. If $\bm{R}:= \bm{\Psi}_\natural \bm{Z}_\natural^\top$ has full row rank and $\bm{\Lambda}^{-1}\bm{Z}_\natural^\top \mathcal{B}_\infty$ is sufficiently scattered, then for any $(\bm{\Psi}_\star, \bm{Z}_\star)$ such that
        \begin{equation}
            (\bm{\Psi}_\star, \bm{Z}_\star) \in \arg\min_{\bm{\Psi}, \bm{Z}} \vol(\bm{\Psi}) \text{ s.t. } \begin{cases}
                \bm{R} = \bm{\Psi}\bm{Z}^\top \\
                \max_{i \in [T]} \|\bm{Z}_i\|_1 \le 1,
            \end{cases}
        \end{equation}
        we have that $(\bm{\Psi}_\star, \bm{Z}_\star) \cong (\bm{\Psi}_\natural, \bm{Z}_\natural)$.
    \end{theorem}
    
\section{Blind Identification by Volume Minimization}
    In this work, we consider the problem of identifying a ground truth system and its unknown inputs $(\bm{A}_\natural, \bm{B}_\natural, (\bm{u}_\natural^{(k)})_{k \in [T]}) \in \mathbb{R}^{n \times n} \times \mathbb{R}^{n \times m} \times (\mathbb{R}^{m})^T$, where $\bm{B}_\natural$ has full column rank ($n \ge m$), from $T+1$ state observations $(\bm{x}^{(k)})_{k=0}^T \in (\mathbb{R}^n)^{T+1}$, where $\bm{x}^{(k+1)} = \bm{A}_\natural\bm{x}^{(k)} + \bm{B}_\natural\bm{u}_\natural^{(k)}$.
    Throughout, we make use of the matrix notation $\bm{X}_+ \in \mathbb{R}^{n \times T}, \bm{X} \in \mathbb{R}^{T \times n}$ and $\bm{U} \in \mathbb{R}^{T \times m}$ as defined in the introduction.
    To these ends, for any feasible system parameters and inputs $(\bm{A}, \bm{B}, \bm{U})$, the input-state relationship may be compactly expressed over all $k \in [T]$ as $\bm{X}_+ = \bm{A}\bm{X}^\top + \bm{B}\bm{U}^\top$.
    Inspired by the volume formulation of complete dictionary learning considered in the work of Hu and Huang \cite{hu2023global}, we consider
    \begin{equation}\tag{P}\label{opt:L1_constr}
        \min_{\bm{A}, \bm{B}, \bm{U}} \text{vol} (\bm{B}) \text{ s.t. } \begin{cases}
            \bm{X}_+ = \bm{A}\bm{X}^\top + \bm{B}\bm{U}^\top \\
            \max_{i} \|\bm{U}_i\|_1 \le T\mu
        \end{cases}
    \end{equation}
    where $\mu$ is a parameter chosen to fix the average entry  of the learned input representation. While we argue for its utility beyond this case, this formulation naturally emerges as a maximum likelihood estimator:

    \begin{remark}\label{rmk:mle}
        The optimization problem \eqref{opt:L1_constr} is the MLE for $(\bm{A}, \bm{B})$ given $(\bm{x}^{(k)})_{k = 0 }^T$, under $\bm{u}_i^{(k)} \overset{\text{i.i.d.}}{\sim} \text{Laplace}(0, \mu)$.
    \end{remark}
    \begin{proof}
        See appendix (\ref{proof:mle}).
    \end{proof}

    When the matrix $\bm{A}_\natural$ is known and $m=n$, this problem reduces to the complete dictionary learning formulation considered by \cite{hu2023global}. In that setting, the geometric sufficient scattering of $\tilde{\bm{U}}_\natural^\top \mathcal{B}_\infty$, where $\tilde{\bm{U}}_\natural$ is the $\ell_1$ column-normalized version of $\bm{U}_\natural$, provides a certificate of global identifiability. In this work, we develop an analogous condition for the case of a linear system with unknown $\bm{A}_\natural$ and $\bm{B}_\natural$ of full column rank, on the following \textit{fundamental zonotope}:
    \begin{definition}[Fundamental Zonotope]
        Define the $\ell_1$ column-normalized matrices $\tilde{\bm{X}}, \tilde{\bm{U}}$ such that $\tilde{\bm{X}}_i = \bm{X}_i/\|\bm{X}_i\|_1, \tilde{\bm{U}}_i = \bm{U}_i/\|\bm{U}_i\|_1$. We define the \textit{fundamental zonotope} $\mathcal{Z}$ of $(\bm{x}, \bm{u})$ as 
        \begin{equation}
            \mathcal{Z} := \begin{bmatrix}
                \tilde{\bm{X}} & \tilde{\bm{U}}
            \end{bmatrix}^\top \mathcal{B}_\infty \subseteq \mathbb{R}^{n} \times \mathbb{R}^m.
        \end{equation}
    \end{definition}
    This defines a zonotope of dimension $n + m$ if and only if $(\bm{x}, \bm{u})$ is persistently exciting. Therefore, our primary geometric condition for identifiability presents as a strengthening of this property:
    \begin{definition}[Persistently Scattering]
        We say that $(\bm{x}, \bm{u})$ is \textit{persistently scattering} if there exists a nonsingular diagonal matrix $\bm{D}$ such that 
        \begin{equation}
            \bm{D} \mathcal{B}_2 \subseteq \mathcal{Z}, \quad \bm{D} \mathcal{B}_2 \cap \partial \mathcal{Z} = \{(0, \pm\bm{e}^{(i)})\}_{i \in [m]}.
        \end{equation}
    \end{definition}
    We may remark that this is weaker than sufficient scattering of $\mathcal{Z}$. However, note that $\mathcal{Z} \cap (\{0\} \times \mathbb{R}^{m}) = \{0\} \times \mathcal{U}$, where $\mathcal{U} := \tilde{\bm{U}}^\top (\mathrm{img}(\tilde{\bm{X}})^\perp \cap \mathcal{B}_\infty)$. When the persistent scattering condition is satisfied, it implies sufficient scattering of the set $\mathcal{U}$--though it is strictly stronger, as sufficient scattering of this intersection does not require a full-dimensional ellipse to be inscribed in $\mathcal{Z}$.

    We will now establish an equivalent characterization of persistent scattering in terms of the support function of $\mathcal{Z}$:
    \begin{equation}
        \sigma_\mathcal{Z}(\bm{w}, \bm{v}) := \sup_{(\bm{a}, \bm{b}) \in \mathcal{Z}} \langle \bm{a}, \bm{w} \rangle + \langle \bm{b}, \bm{v} \rangle =  \|\tilde{\bm{X}}\bm{w} + \tilde{\bm{U}}\bm{v} \|_1.
    \end{equation}

    \begin{proposition}[Support Function Characterization]\label{prop:support_equiv}
        Define $\Omega(\bm{v}) := \inf_{\bm{w}} \sigma_{\mathcal{Z}}(\bm{w}, \bm{v})$.
        The pair $(\bm{x}, \bm{u})$ is persistently scattering if and only if the following hold simultaneously:
        \begin{enumerate}
            \item $\Omega(\bm{v}) \ge \| \bm{v}\|_2$ \label{eq:norm_ineq}
            \item $\Omega(\bm{v}) = \|\bm{v}\|_2 \iff \|\bm{v}\|_0 \le 1$ \label{eq:norm_eq}
            \item $\forall i \in [m], \left(\Omega(\bm{e}^{(i)}) = \sigma_{\mathcal{Z}}(\bm{w}, \bm{e}^{(i)}) \implies \bm{w} = 0\right)$ \label{eq:norm_nsp}
        \end{enumerate}
    \end{proposition}
    \begin{proof}
        See appendix \ref{proof:support_equiv}. The proof is nontrivial and relies heavily on compactness of the unit sphere, and piecewise linearity of $\sigma_{\mathcal{Z}}(\bm{w}, \bm{v})$ and $\Omega(\bm{v})$.
    \end{proof}

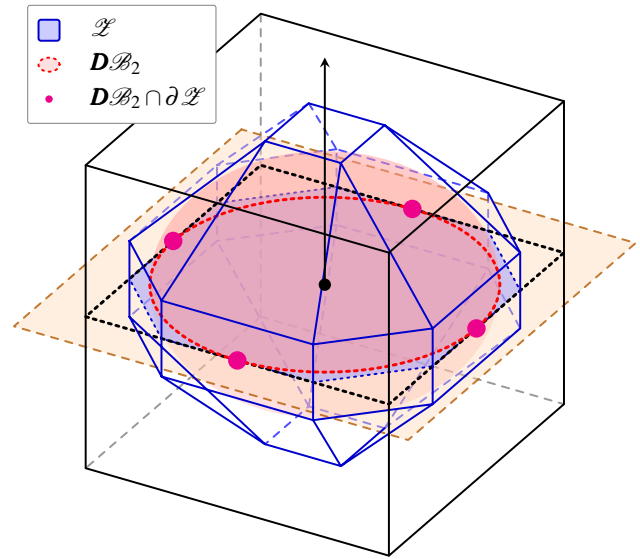
\begin{figure}[t]
    \centering
    \resizebox{\columnwidth}{!}{%
    \tdplotsetmaincoords{60}{120}
    
    \begin{tikzpicture}[tdplot_main_coords, scale=2.5, line join=round, line cap=round]
    
        \pgfmathsetmacro{\b}{0.5}
        \pgfmathsetmacro{\c}{0.25}
        \pgfmathsetmacro{\eY}{sqrt(0.25 + 0.75 * (4/9))}
    
        \coordinate (V1) at (1, \b, 0);
        \coordinate (V2) at (\b, 1, 0);
        \coordinate (V3) at (-\b, 1, 0);
        \coordinate (V4) at (-1, \b, 0);
        \coordinate (V5) at (-1, -\b, 0);
        \coordinate (V6) at (-\b, -1, 0);
        \coordinate (V7) at (\b, -1, 0);
        \coordinate (V8) at (1, -\b, 0);
    
        \draw[dashed, thick, gray!80] (-1,-1,-1) -- (1,-1,-1);
        \draw[dashed, thick, gray!80] (-1,-1,-1) -- (-1,1,-1);
        \draw[dashed, thick, gray!80] (-1,-1,-1) -- (-1,-1,1);
    
        \draw[dashed, thick, blue!70] 
            (-1,\b,\c) -- (-1,-\b,\c) -- (-1,-\b,-\c) -- (-1,\b,-\c) -- cycle
            (\b,-1,\c) -- (-\b,-1,\c) -- (-\b,-1,-\c) -- (\b,-1,-\c)
            (-\c,\c,-1) -- (-\c,-\c,-1) -- (\c,-\c,-1)
            (-1,\b,-\c) -- (-\b,1,-\c)
            (-1,-\b,\c) -- (-\b,-1,\c)
            (-1,-\b,-\c) -- (-\b,-1,-\c)
            (-\c,-\c,1) -- (-1,-\b,\c)
            (-\c,-\c,1) -- (-\b,-1,\c)
            (-\c,\c,-1) -- (-1,\b,-\c)
            (-\c,-\c,-1) -- (-1,-\b,-\c)
            (-\c,\c,-1) -- (-\b,1,-\c)
            (-\c,-\c,-1) -- (-\b,-1,-\c);

        \begin{scope}[tdplot_screen_coords]
            \fill[red!40, opacity=0.3] (-1,0) arc (180:360:1 and \eY) arc (0:180:1 and 0.5);
        \end{scope}
    
        \fill[orange!30, opacity=0.4] (-1.3, -1.3, 0) -- (1.3, -1.3, 0) -- (1.3, 1.3, 0) -- (-1.3, 1.3, 0) -- cycle;
        \draw[orange!70!black, thick, dashed, opacity=0.8] (-1.3, -1.3, 0) -- (1.3, -1.3, 0) -- (1.3, 1.3, 0) -- (-1.3, 1.3, 0) -- cycle;
    
        \fill[blue!40, opacity=0.5] (V1) -- (V2) -- (V3) -- (V4) -- (V5) -- (V6) -- (V7) -- (V8) -- cycle;
        \draw[dotted, thick, blue!80!black] (V1) -- (V2) -- (V3) -- (V4) -- (V5) -- (V6) -- (V7) -- (V8) -- cycle;
    
        \begin{scope}[tdplot_screen_coords]
            \fill[red!40, opacity=0.5] (1,0) arc (0:180:1 and \eY) arc (180:360:1 and 0.5);
        \end{scope}
    
        \draw[dotted, very thick, black] (1,-1,0) -- (1,1,0) -- (-1,1,0) -- (-1,-1,0) -- cycle;
        
        \draw[dotted, very thick, red] plot[domain=0:360, samples=73] ({cos(\x)}, {sin(\x)}, 0);
    
        \draw[thick, blue!80!black]
            (1,\b,\c) -- (1,-\b,\c) -- (1,-\b,-\c) -- (1,\b,-\c) -- cycle
            (\b,1,\c) -- (-\b,1,\c) -- (-\b,1,-\c) -- (\b,1,-\c) -- cycle
            (\c,\c,1) -- (-\c,\c,1) -- (-\c,-\c,1) -- (\c,-\c,1) -- cycle
            (\b,-1,-\c) -- (\b,-1,\c)
            (\c,\c,-1) -- (-\c,\c,-1)
            (\c,-\c,-1) -- (\c,\c,-1)
            (1,\b,\c) -- (\b,1,\c)
            (1,\b,-\c) -- (\b,1,-\c)
            (1,-\b,\c) -- (\b,-1,\c)
            (1,-\b,-\c) -- (\b,-1,-\c)
            (-1,\b,\c) -- (-\b,1,\c)
            (\c,\c,1) -- (1,\b,\c)
            (\c,-\c,1) -- (1,-\b,\c)
            (-\c,\c,1) -- (-1,\b,\c)
            (\c,\c,1) -- (\b,1,\c)
            (-\c,\c,1) -- (-\b,1,\c)
            (\c,-\c,1) -- (\b,-1,\c)
            (\c,\c,-1) -- (1,\b,-\c)
            (\c,-\c,-1) -- (1,-\b,-\c)
            (\c,\c,-1) -- (\b,1,-\c)
            (\c,-\c,-1) -- (\b,-1,-\c);
    
        \fill (0,0,0) circle (1pt);
        \draw[-stealth, thick] (0,0,0) -- (0,0,1.5);
    
        \fill[magenta] (1, 0, 0) circle (1.5pt);
        \fill[magenta] (0, 1, 0) circle (1.5pt);
        \fill[magenta] (-1, 0, 0) circle (1.5pt);
        \fill[magenta] (0, -1, 0) circle (1.5pt);
    
        \draw[thick] (1,-1,-1) -- (1,1,-1) -- (-1,1,-1) -- (-1,1,1) -- (-1,-1,1) -- (1,-1,1) -- cycle;
        
        \draw[thick] (1,1,1) -- (1,1,-1);
        \draw[thick] (1,1,1) -- (-1,1,1);
        \draw[thick] (1,1,1) -- (1,-1,1);
    
        \begin{scope}[tdplot_screen_coords]
            \node[draw=gray, fill=white, rounded corners=1pt, anchor=north west, inner sep=4pt] at (-1.7, 1.6) {
                \begin{tabular}{@{}cl@{}}
                    \tikz[baseline=-0.5ex]{\draw[blue!80!black, thick, fill=blue!40, fill opacity=0.5] (-0.15,-0.15) rectangle (0.15,0.15);} & $\mathcal{Z}$ \\[2pt]
                    \tikz[baseline=-0.5ex]{\draw[red, dotted, thick, fill=red!40, fill opacity=0.3] (0,0) ellipse (0.15 and 0.1);} & $\bm{D} \mathcal{B}_2$ \\[2pt]
                    \tikz[baseline=-0.5ex]{\fill[magenta] (0,0) circle (1.5pt);} & $\bm{D} \mathcal{B}_2 \cap \partial \mathcal{Z}$
                \end{tabular}
            };
        \end{scope}
    
    \end{tikzpicture}%
    } 
    \vspace{-1em}
    \caption{A visualization of the persistent scattering condition, viewing $\mathbb{R}^n \times \mathbb{R}^m \to \mathbb{R} \times \mathbb{R}^2$, with the first (states) coordinate vertical, and remaining coordinates (inputs) in the $x, y$ plane. The set $\mathcal{U}$, inscribed in $\{0\} \times [-1, 1]^2$, is sufficiently scattered relative to the $x, y$ plane. The set $\mathcal{Z}$ (the polyhedron outlined in solid blue) is not sufficiently scattered, but contains an ellipse (in red) such that the intersection of its boundary with the ellipse is $\{(0, \pm1, 0), (0, 0, \pm 1)\}$.}
    \label{fig:polyhedron_intersection}
\end{figure}

    \begin{table*}[ht]\label{tab:conditions}
    
\centering
\caption{%
    Comparison of identifiability conditions. $\tilde{\bm{U}} \in \mathbb{R}^{T \times m}$ and $\tilde{\bm{X}} \in \mathbb{R}^{T \times n}$ denote $\ell_1$ column-normalized input and state matrices respectively. $\mathcal{Z} := [\tilde{\bm{X}}, \tilde{\bm{U}}]^\top \mathcal{B}_\infty$. $\sigma_\mathcal{Z}$ denotes the support function of $\mathcal{Z}$, and $\Omega(\bm{v}) := \inf_{\bm{w}} \sigma_{\mathcal{Z}}(\bm{w}, \bm{v})$
}

\newcolumntype{L}[1]{>{\raggedright\arraybackslash}m{#1}}

\renewcommand{\arraystretch}{1.55}
\begin{tabular}{@{} L{2cm} L{2.5cm} L{3.5cm} L{3.5cm} L{2.5cm} @{}}
\toprule
\textbf{Condition} & \textbf{Setting} & \textbf{Geometric Characterization}
    & \textbf{Norm Characterization} & \textbf{Recovers} \\
\midrule
Sufficient \mbox{Scattering}~\cite{hu2023global}
    & Dict.\ learning \newline (known $\bm{A}_\natural$, $m \le n$)
    & $\mathcal{B}_2 \subseteq \tilde{\bm{U}}_\natural^\top \mathcal{B}_\infty$
      and
      \mbox{$\mathcal{B}_2 \cap \partial(\tilde{\bm{U}}_\natural^\top \mathcal{B}_\infty)
       = \{\pm\bm{e}_i\}_{i\in[m]}$}
    & $\|\tilde{\bm{U}}_\natural\bm{v}\|_1
       \ge \|\bm{v}\|_2\;\forall\,\bm{v}$,
      \mbox{with equality $\Rightarrow \|\bm{v}\|_0 \le 1$}
    & $(\bm{B}_\natural, \bm{U}_\natural)$ up to
      equivalence \\[4pt]
Persistent \mbox{Excitation}~\cite{willems-et-al-scl05}
    & Known-$\bm{U}_\natural$ sysid
    & $\exists \bm{D} \succ 0 \text{ diagonal}, \bm{D} \mathcal{B}_2 \subseteq \mathcal{Z}$
    & $\exists \varepsilon > 0, \sigma_{\mathcal{Z}}(\bm{w}, \bm{v}) \ge \varepsilon\|(\bm{w}, \bm{v})\|_2$
    & $(\bm{A}_\natural, \bm{B}_\natural)$ exactly \\[4pt]
Persistent Scattering [Ours]
    & Blind sysid \newline ($\bm{A}_\natural, \bm{B}_\natural, \bm{U}_\natural$ unknown)
    & $\exists \bm{D} \succ 0 \text{ diagonal }$, \mbox{$\bm{D}\mathcal{B}_2 \subseteq \mathcal{Z}$}, \mbox{$\bm{D}\mathcal{B}_2 \cap \partial \mathcal{Z} = \{(0, \pm \bm{e}_i)\}_{i \in [m]}$}
    & $\Omega(\bm{v}) \ge
       \|\bm{v}\|_2$,
      \mbox{$\Omega(\bm{v}) = \|\bm{v}\|_2 \implies \|\bm{v}\|_0 \le 1$,}
      \mbox{$\Omega(\bm{v}) = \sigma_{\mathcal{Z}}(\bm{w}, \bm{v}) \implies \bm{w} = 0$}
    & $\bm{A}_\natural$ exactly; $(\bm{B}_\natural, \bm{U}_\natural)$
      up to equivalence \\
\bottomrule
\end{tabular}
    \end{table*}

    We note that the gauge $\Omega(\bm{v})$ is exactly the support function of $\mathcal{U}$. Therefore, the first two conditions of proposition \ref{prop:support_equiv} concern the behavior of this intersection, and mirror the sufficient scattering characterization in \cite{hu2023global}. The third manifests due to the additional ambiguity introduced by the unknown $\bm{A}$ matrix, reflecting the need for the zonotope $\mathcal{Z}$ to be nondegenerate.
    
    Our main theorem states that the persistent scattering condition, which we leverage in the proof through its characterization in proposition \ref{prop:support_equiv}, is sufficient for recovery of the system and unknown inputs, up to equivalence with $(\bm{B}_\natural, \bm{U}_\natural)$.

    \begin{theorem}\label{thm:main}
        Suppose that $\bm{X}_+ = \bm{A}_\natural\bm{X}^\top + \bm{B}_\natural \bm{U}_\natural^\top$, where $\ker(\bm{B}_\natural) =0$, and $(\bm{x}, \bm{u}_\natural)$ is persistently scattering. Then every globally optimal solution $(\bm{A}_\star, \bm{B}_\star, \bm{U}_\star)$ to \eqref{opt:L1_constr} satisfies $\bm{A}_\star = \bm{A}_\natural$ and $(\bm{B}_\star, \bm{U}_\star) \cong (\bm{B}_\natural, \bm{U}_\natural)$.
    \end{theorem}
 
    \begin{proof}
        Denote $\bm{P}_{X}^\perp = \bm{I} - \bm{X}\bm{X}^\dagger$. By lemma \ref{lemma:equiv_cond}, it suffices to study
        \[
            \mathcal{G} := \arg\min_{\bm{A}, \bm{B}, \bm{U}, \bm{\Phi}} \vol(\bm{B}) \text{ s.t. } \begin{cases}
                \bm{A} = (\bm{X}_+ - \bm{B}\bm{U}^\top)(\bm{X}^\top)^\dagger \\
                \bm{P}_{X}^\perp\bm{U} = \bm{P}_{X}^\perp\bm{U}_\natural \bm{\Phi} \\
                \bm{B}_\natural = \bm{B}\bm{\Phi}^\top \\
                \max_{i} \|\bm{U}_i\|_1 \le T\mu
            \end{cases}
        \]
        Let $(\bm{A}_\star, \bm{B}_\star, \bm{U}_\star, \bm{\Phi}_\star) \in \mathcal{G}$.
        By the full rank of $\bm{B}_\natural$, it follows that $\bm{\Phi}_\star$ is invertible. 
        By the definition of the volume and monotonicity of the logarithm, it follows that for any feasible $\bm{B}, \bm{\Phi}$, $\log \vol(\bm{B}) = \frac{1}{2}\log\det(\bm{\Phi}^{-1}\bm{B}_\natural^\top\bm{B}_\natural \bm{\Phi}^{-\top}) = \frac{1}{2}\log\det(\bm{B}_\natural^\top\bm{B}_\natural) - \log|\bm{\Phi}|$,
        and so
        \begin{align*}
            \bm{A}_\star &= (\bm{X}_+ - \bm{B}_\star \bm{U}_{\star}^\top)\bm{X}^\dagger, \quad \bm{B}_\star = \bm{B}_\natural \bm{\Phi}_\star^{-\top} \\(\bm{U}_\star, \bm{\Phi}_\star) &\in \arg\min_{\bm{U}, \bm{\Phi}} -\log|\bm{\Phi}| \text{ s.t. } \forall i, \begin{cases}
                \bm{P}_{X}^\perp\bm{U}_i = \bm{P}_{X}^\perp\bm{U}_\natural\bm{\Phi}_i \\
                \|\bm{U}_i\|_1 \le T\mu
            \end{cases}
        \end{align*}
        Observe that
        \begin{align*}
            \Omega(\bm{\phi})
            &= \inf_{\bm{w}} \|\tilde{\bm{X}}\bm{w} + \tilde{\bm{U}_\natural}\bm{\phi}\|_1 \\
            &= \inf_{\bm{v}, \bm{w}} \|\bm{v}\|_1 \text{ s.t. } \bm{v} = \tilde{\bm{X}}\bm{w} + \tilde{\bm{U}}_\natural \bm{\phi} \\
            &= \inf_{\bm{v}} \|\bm{v}\|_1 \text{ s.t. } \bm{P}_{X}^\perp\bm{v} = \bm{P}_{X}^\perp\bm{U}_\natural(\bm{\omega}^{-1} \odot \bm{\phi})
        \end{align*}
        where we have defined $\bm{\omega}_i := \|(\bm{U}_\natural)_i\|_1$.
        It follows that
        \[
            \bm{\Phi}_\star \in \arg\max_{\bm{\Phi}} |\bm{\Phi}| \text{ s.t. } \forall i, \Omega(\bm{\omega} \odot\bm{\Phi}_i) \le T\mu
        \]
        By Hadamard's inequality and persistent scattering C\ref{eq:norm_ineq}, we have
        \begin{align*}
            |\bm{\omega} \odot \bm{\Phi}| \le \prod_{i \in [m]} \|\bm{\omega} \odot \bm{\Phi}_i\|_2 \le \prod_{i \in [m]} \Omega(\bm{\omega} \odot \bm{\Phi}_i) \le (T\mu)^m,
        \end{align*}
        where each inequality is tight if and only if the columns of $\bm{\Phi}$ are orthogonal (equality for Hadamard's inequality) and 1-sparse (by persistent scattering C\ref{eq:norm_eq}). Since $|\bm{\Phi}| = |\bm{\omega} \odot \bm{\Phi}|/\prod_{i \in [m]} \|\bm{U}_i\|_1$, we may conclude that these inequalities must also be tight for $\bm{\Phi}_\star$. Therefore, $\bm{\Phi}_\star = \bm{D}\bm{\Pi}$ for some nonsingular diagonal $\bm{D}$ and permutation matrix $\bm{\Pi}$, where in particular we may take $\bm{D}_{ii} = T\mu/\|(\bm{U}_\natural)_i\|_1$. The optimal $(\bm{U}_\star)_i$ therefore equals $\tilde{\bm{X}}\bm{w}^* + \bm{U}_\natural(\bm{\Phi}_\star)_i$, where $\bm{w}^*$ minimizes $\|\tilde{\bm{X}}\bm{w}^* + \bm{U}_\natural(\bm{\Phi}_\star)_i\|_1$.
        By persistent scattering C\ref{eq:norm_nsp}, the sparsity of $(\bm{\Phi}_\star)_i$ implies that $\bm{w}^* =0$, so $\bm{U}_\star = \bm{U}_\natural \bm{\Phi}_\star = \bm{U}_\natural \bm{D}\bm{\Pi}$.
        We conclude that $\bm{B}_\star = \bm{B}_\natural \bm{D}^{-1}\bm{\Pi}$, and lastly, as $\bm{B}_\star\bm{U}_\star^\top = \bm{B}_\natural \bm{\Phi}_\star^{-\top} \bm{\Phi}_\star^\top \bm{U}_\natural^\top = \bm{B}_\natural \bm{U}_\natural^\top$, $\bm{A}_\star = (\bm{X}_+ - \bm{B}_\star \bm{U}_\star)(\bm{X}^\top)^\dagger = \bm{A}_\natural$.
    \end{proof}

\section{BLISS: Algorithmic Implementation}
In this section we detail the BLISS algorithm (Blind Identification of Sparse Signals), our algorithmic approach to solving the problem \eqref{opt:L1_constr}. Broadly speaking, we develop a generalized alternating direction method of multipliers, with primal-dual residual balancing, leveraging efficiently computable proximal operators of $f(\bm{\Phi}) = -\log|\bm{\Phi}|$ and projections onto the scaled $\ell_1$ ball \cite{condat2016fast}. This is in contrast to prior work on complete dictionary learning via volume optimization, which employed a linearized ADMM approach. Crucially, we employ the following equivalence:

\begin{proposition}\label{prop:svd_opt}
    Let $\bm{P}_{\bm{X}}^\perp := \bm{I} - \bm{X}\bm{X}^\dagger$, and $\bm{X}_+ \bm{P}_{\bm{X}}^\perp = \bm{Q}\bm{\Sigma}\bm{V}^\top$ be a thin SVD\footnote{\textit{A thin SVD} is taken to be such that $\bm{\Sigma}$ is square and full rank.}. Denote by $\mathcal{G} :=\{(\bm{\Phi}_\star^{(i)}, \bm{U}_\star^{(i)})\}_{i \in I}$ the set of global optima of
    \[
        \min_{\bm{\Phi}, \bm{U}} -\log|\bm{\Phi}| \text{ s.t. } \begin{cases}
            \bm{V}\bm{\Phi} = \bm{P}_{\bm{X}}^\perp\bm{U} \\ 
            \max_{i}\|\bm{U}_i\|_1 \le \mu T
        \end{cases}
    \]

    The set of global optima of \eqref{opt:L1_constr} satisfies $\mathcal{G}_P = \{((\bm{X}_+ - \bm{Q}\bm{\Sigma}(\bm{\Phi}_\star^{(i)})^{-\top}(\bm{U}_\star^{(i)})^\top)(\bm{X}^\top)^\dagger, \bm{Q}\bm{\Sigma}(\bm{\Phi}_\star^{(i)})^{-\top}, \bm{U}_\star^{(i)}\}_{i \in I}$.
\end{proposition}
\begin{proof}
    See appendix.
\end{proof}

\begin{figure*}[ht]
    \centering
    \resizebox{\textwidth}{!}{\input{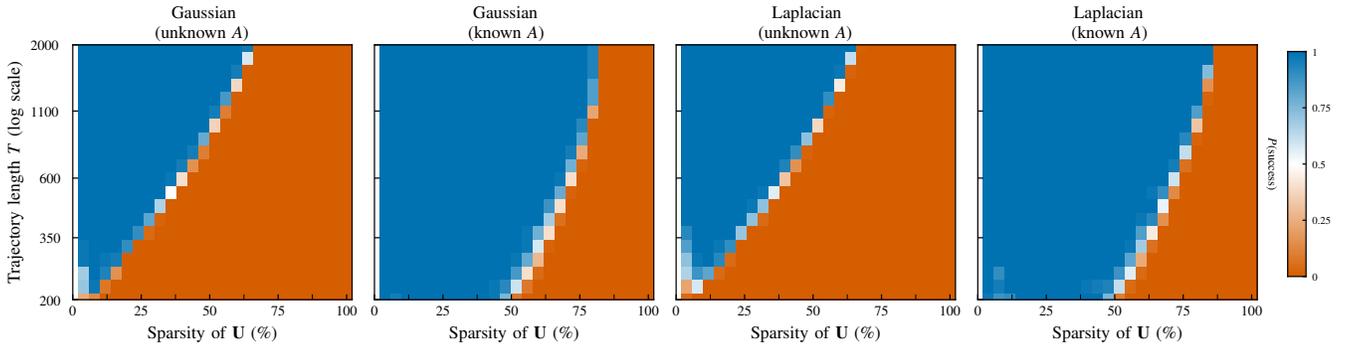}}
    \caption{Probability of successful recovery ($P(\text{success})$) as a function
        of sparsity fraction $s/m$ and trajectory length $T$
        for a  stable, and controllable system with $n=100$, $m=25$, and 50 trials per cell.}
    \label{fig:2d_heatmaps}
\end{figure*}

It is based on this form that we develop our algorithm. Denote the extended-real valued functions
\[
    f(\bm{\Phi}) = \begin{cases}
        -\log|\bm{\Phi}|,  & |\bm{\Phi}| \ne 0 \\
        \infty, & \text{o.w.}
    \end{cases}, \,\, \delta_{\mathcal{C}} = \begin{cases}
        0, & \max_{i} \|\bm{U}_i\|_1 \le T\mu \\
        \infty, & \text{o.w.}
    \end{cases}
\]
The standard augmented Lagrangian reads
\begin{align}
\begin{split}
    \mathcal{L}_\rho(\bm{U}, \bm{\Phi}, \bm{\Lambda}) &= f(\bm{\Phi}) + \delta_{\mathcal{C}}(\bm{U}) \\ &+ \langle \bm{\Lambda}, \bm{V}\bm{\Phi} - \bm{P}_{\bm{X}}^\perp\bm{U} \rangle + \frac{\rho}{2}\|\bm{V}\bm{\Phi} - \bm{P}_{\bm{X}}^\perp\bm{U} \|_F^2.
\end{split}
\end{align}
Under the regularization schedule $\rho^{(k)}$, we employ a Gauss-Seidel generalized ADMM, leveraging an additional proximal term to decouple the $\bm{U}$ update from the geometry of $\mathrm{img}(\bm{X})$:
\begin{align}
\begin{split}
    \bm{U}^{(k+1)} &= \arg\min_{\bm{U}} \mathcal{L}_{\rho^{(k)}} (\bm{U}, \bm{\Phi}^{(k)}, \bm{\Lambda}^{(k)}) \\&+ \frac{\rho^{(k)}}{2}\|(\bm{I} - \bm{P}_{\bm{X}}^\perp)(\bm{U} - \bm{U}^{(k)})\|_F^2 \\
    \bm{\Phi}^{(k+1)} &= \arg\min_{\bm{\Phi}} \mathcal{L}_{\rho^{(k)}}(\bm{U}^{(k+1)}, \bm{\Phi}, \bm{\Lambda}^{(k)}) \\
    \bm{\Lambda}^{(k+1)} &= \bm{\Lambda}^{(k)} + \rho^{(k)}(\bm{V}\bm{\Phi}^{(k+1)} - \bm{P}_{\bm{X}}^\perp\bm{U}^{(k+1)})
\end{split}
\end{align}
This results in the closed-form updates 
\begin{align}\label{eq:bliss_updates}
\begin{split}
    \bm{U}^{(k+1)}  &= \Pi_{\mathcal{C}}\left( \bm{V}\bm{\Phi}^{(k)} + \bm{\Lambda}^{(k)}/\rho^{(k)} + (\bm{I} - \bm{P}_{\bm{X}}^\perp) \bm{U}^{(k)} \right) \\
    \bm{\Phi}^{(k+1)} &= \mathrm{prox}_{f/\rho^{(k)}}\left( \bm{P}_{\bm{X}}^\perp\bm{U}^{(k+1)} - \bm{\Lambda}^{(k)}/\rho^{(k)} \right)
\end{split}
\end{align}
The projection $\Pi_\mathcal{C}$ onto the scaled $\ell_1$-ball is implemented column-wise. The proximal operator of $f$ can be evaluated in closed form in terms of the SVD \cite{parikh2014proximal}:
\begin{remark}
    Let $\bm{W} = \bm{U}_W \bm{\Sigma}_W \bm{V}_W^\top$ be an SVD. Then
    \[
        \mathrm{prox}_{\lambda f}(\bm{W}) = \frac{1}{2}\bm{U}_W \left( \bm{\Sigma}_W + \sqrt{\bm{\Sigma}_W^2 + 4\lambda\bm{I}}  \right)\bm{V}_W^\top.
    \]
\end{remark}
Within this update scheme, we denote the primal residual $\bm{R}^{(k)} := \bm{V}\bm{\Phi}^{(k)} - \bm{P}_{\bm{X}}^\perp \bm{U}^{(k)}$ and dual residual $\bm{S}^{(k)} := \rho^{(k)} \bm{V}(\bm{\Phi}^{(k+1)} - \bm{\Phi}^{(k)})$.
The $\rho^{(k)}$ schedule is implemented as the classical primal-dual residual balancing scheme detailed in \cite{boyd2011distributed} 
via an adaptation threshold $\alpha > 0$, and factor, $\tau > 0$:
\begin{align}\label{eq:adaptive_factor}
    \rho^{(k+1)} = \begin{cases}
        \tau \rho^{(k)}, & \|\bm{R}^{(k)}\|_F \ge \alpha \|\bm{S}^{(k)}\|_F \\
        \tau^{-1} \rho^{(k)}, & \|\bm{S}^{(k)}\|_F \ge \alpha \|\bm{R}^{(k)}\|_F \\ 
        \rho^{(k)}, & \text{ otherwise}
    \end{cases}
\end{align}
We note as well that there is a natural floor on admissible values of $\rho$, which arises from the proximal operator of $f$:
\begin{corollary}\label{crl:rho_bound}
    $\forall k>0, \sigma_{\min}(\bm{\Phi}^{(k)}) \ge 1/\sqrt{\rho^{(k)}}$
\end{corollary}
Therefore, to recover $(\bm{B}_\natural, \bm{U}_\natural)$ where $\|(\bm{U}_{\natural})_i\|_1=\mu T$, $\rho^{(k)}$ must satisfy $\lim \inf_{k} \rho^{(k)} \ge  \sigma_{\max}(( \bm{X}_+\bm{P}_{\bm{X}}^\perp)^\dagger\bm{B}_\natural)^2$.
\section{Experiments}

In this section, we implement the BLISS algorithm.
\paragraph{Data generation} 
We consider a system with state dimension $n = 100$ and input 
dimension $m = 25$. For each $t \in [T]$, we choose a random set 
$\mathcal{S} \subseteq [m]$ without replacement such that 
$|\mathcal{S}| = s$. In an i.i.d.\ manner, for each $j \in 
\mathcal{S}$ we sample $\mathbf{u}^{(t)}_{j} \sim \text{Gaussian} 
\text{ / Laplace}$. The matrix $\bm{A}_\natural$ is generated by scaling an i.i.d. Gaussian random matrix to have spectral radius $0.9$.
The matrix $\bm{B}_{\natural}$ is chosen to be a standard i.i.d. Gaussian random matrix.
We generate each system trajectory according to \eqref{eq:LTI_eqs} from an initial state $\bm{x}_0=0$.

\paragraph{Algorithm} 
We implement the BLISS algorithm in two variants: (a) direct updates from Equation~\ref{eq:bliss_updates} and (b) when we have oracle access to state transition matrix $\bm{A}_{\natural}$, a.k.a. complete dictionary learning. In this case, the proximal updates would resemble Equation \ref{eq:bliss_updates} except that the projection operator $\bm{P}_{\bm{X}}^\perp$ would be identity. 

Throughout the experiments, we 
set 
$\rho^{(k)} = 2\sigma_{\max}((\bm{X}_+\bm{P}_{\bm{X}}^\perp)^\dagger \bm{B}_\natural)^2$ for $k \in [200]$ unless specified (see 
Corollary~\ref{crl:rho_bound}), and for $k \geq 200$ we use the adaptive schedule described in Equation~\ref{eq:adaptive_factor} 
with $\alpha = 1.2$, and $\tau = 2$ .
Stopping criterion for the algorithm is 
either maximum run of
3000 iterations or when both primal and dual 
residuals are below $10^{-6}$, whichever is 
earliest.
To gauge recovery, we first normalize the columns of $\bm{B}_\star, \bm{U}_\star$ and $\bm{B}_\natural, \bm{U}_\natural$ such that $\|(\bm{U}_\star)_i\|_1 = \|(\bm{U}_\natural)_j\|_1=T\mu$ with the respective columns of $\bm{B}_\star, \bm{B}_\natural$ scaled reciprocally, and then solve a linear assignment problem on the correlation matrix between columns of $\bm{U}_\star, \bm{U}_\natural$ to determine an optimal permutation matrix $\bm{\Pi}$.
We then declare recovery to be successful if 
\begin{align*}
    \max\left(
    \frac{\|\bm{A}_\star - \bm{A}_\natural\|_F}{\|\bm{A}_\natural\|_F},
    \frac{\|\bm{B}_\star\bm{\Pi} - \bm{B}_\natural\|_F}{\|\bm{B}_\natural\|_F},
    \frac{\|\bm{U}_\star \bm{\Pi} - \bm{U}_\natural\|_F}{\|\bm{U}_\natural\|_F}\right) \le 0.01
\end{align*}
With this definition, we report the empirical probability of 
success of the BLISS algorithm for the data generated with 50 
different seeds.

\paragraph{Settings}
We consider two cases (a) blind system identification setting, 
we recover $(\bm{A}_{\natural}, \bm{B}_{\natural}, 
\bm{U}_{\natural})$ with access only to $\bm{X}$, and (b) dictionary 
learning setting, we recover $(\bm{B}_{\natural}, 
\bm{U}_{\natural})$ with access to $(\bm{X}, \bm{A}_{\natural})$. 
We present results by sweeping the sparsity level $s \in [m]$ 
and the trajectory length $T \in [200, 2000]$.

\paragraph{Discussion}
In Figure~\ref{fig:2d_heatmaps}, we observe clear evidence phase transition behavior mediating exact recovery of the ground truth $(\bm{A}_\natural, \bm{B}_\natural, \bm{U}_\natural)$ in both the known and unknown settings, via the BLISS algorithm. Degradation in recovery performance is sharp as the phase transition is crossed from the sparse and long trajectory regime to the dense and short in all cases considered. In Figure~\ref{fig:rho_convergence}, we observe convergence in all choices of $\rho^{(0)}$ considered; however, choices slightly larger than $\sigma_{\max}((\bm{X}_+ \bm{P}_{\bm{X}}^\perp)^\dagger\bm{B}_{\natural})^2$ (oracle initialization) lead to far more rapid convergence than an initialization of $\rho^{(0)} \sim 1$. For practical implementation with unknown $\bm{B}_\natural$, we suggest setting $\rho^{(0)} \propto \sigma_{\max}((\bm{X}_+ \bm{P}_X^\perp)^\dagger)^2$. Figure~\ref{fig:1d_stable_ctrl} shows that non-zero inputs 
following a Gaussian distribution can slightly tolerate denser 
inputs compared to those following a Laplace distribution. The 
same figure also shows that, as expected, knowing 
$\bf{A}_{\natural}$ reduces the problem to standard dictionary 
learning, which outperforms blind system identification. 

\begin{figure*}[t]
    \centering
    \includegraphics[width=\textwidth]{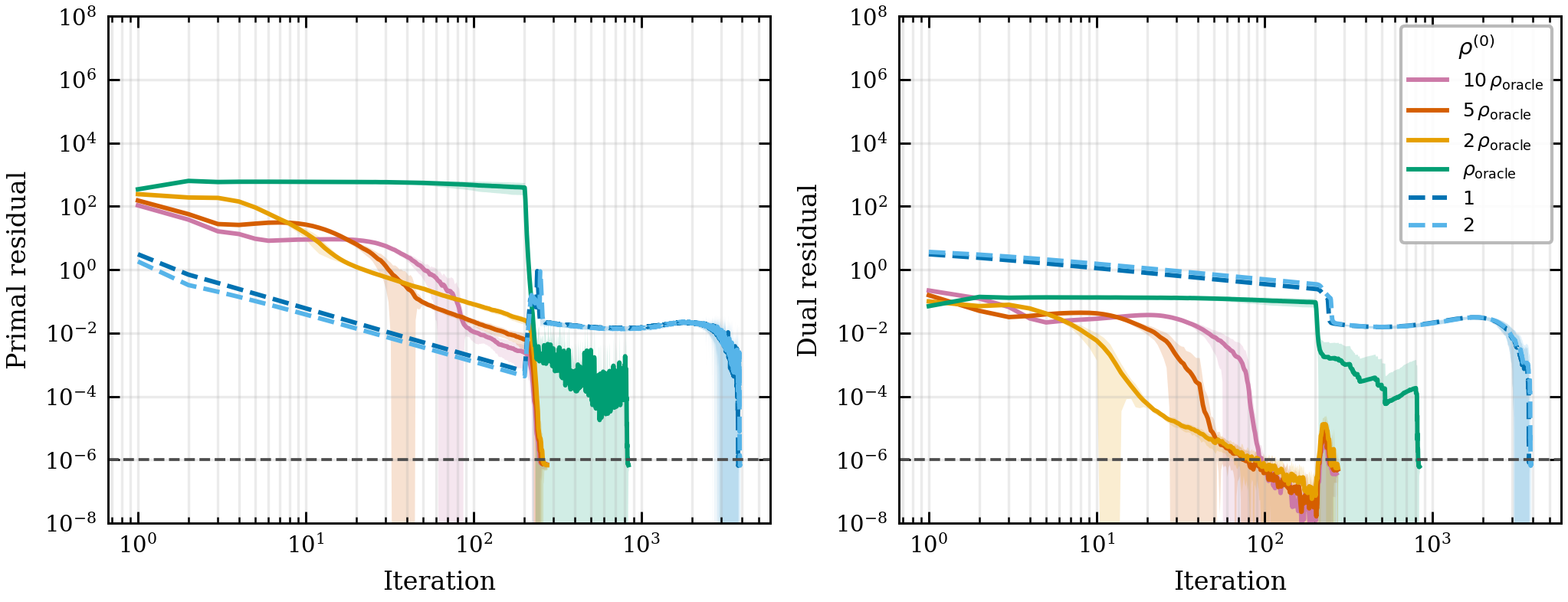}
    \caption{BLISS algorithm convergence for different initializations of the penalty
             parameter $\rho_0$. Solid lines correspond to oracle-scaled
             initializations $\rho^{(0)} = c\rho_{\text{oracle}}$ 
             for $c \in \{1, 2, 5, 10\}$,  where $\rho_{\text{oracle}} = \sigma_{\max}(( \bm{X}_+\bm{P}_{\bm{X}}^\perp)^\dagger \bm{B}_\natural)^2$, while dashed lines show fixed
             initializations $\rho^{(0)} \in \{1, 2\}$.
             Shaded regions indicate $\pm 1$ standard deviation over
             $50$ independent trials. The black dashed horizontal line indicates residual level of 1e-6.}
    \label{fig:rho_convergence}
\end{figure*}

\begin{figure}[h]
    \centering
    \resizebox{\columnwidth}{!}{\input{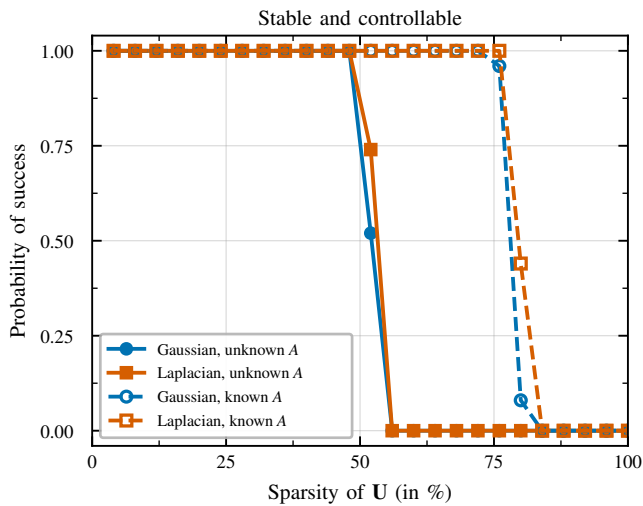}}

    \caption{Probability of successful recovery vs.\ input sparsity $s/m$
             for the stable, controllable case ($n=100$, $m=25$, $T=1000$,
             50 trials per point). Solid lines: unknown $A$ or blind system identification; dashed lines:
             known $A$ (oracle) or dictionary learning.}
    \label{fig:1d_stable_ctrl}
\end{figure}

\section{Conclusions}
In this work, we have formulated the blind identification of LTI systems with sparse inputs as a structured matrix factorization problem and proposed a corresponding algorithm that achieves exact recovery from finitely many samples. Building on connections to complete dictionary learning, we derived an interpretable sufficient condition for global identifiability. Empirically, we demonstrated that a principled application of the alternating direction method of multipliers recovers the ground-truth system across a range of trajectory lengths and sparsity levels. In future work, we aim to extend the volume optimization approach to general LTI systems, and provide a non-asymptotic statistical analysis of the persistent scattering condition. Finally, we plan to investigate the noisy and overcomplete regimes ($m>n$), where additional challenges arise due to analytical tractability of the associated MLE estimator, non-uniqueness, and increased model flexibility.

\bibliographystyle{IEEEtran}
\bibliography{refs_uday}

\appendices
\section{Proofs and Lemmata}
   \begin{proof}(Remark \ref{rmk:mle}) \label{proof:mle}
        We prove the case with $\bm{B}$ assumed square--the general complete $\bm{B} \in \mathbb{R}^{n \times m}$ case proceeds \textit{mutatis mutandis} by considering $p(\bm{x}^{(k+1)} \mid \bm{x}^{(k)}, \bm{A}, \bm{B})$ as a density with respect to the $m$-dimensional Hausdorff measure and replacing $|\det(\bm{B})|$ with $\vol(\bm{B})$ in the standard ICA derivation. Write the innovation process $\bm{r}^{(k)} = \bm{x}^{(k+1)} - \bm{A}\bm{x}^{(k)}$. It follows from classical ML estimation, as in ICA \cite{tharwat2021independent}, that
        \[
            p(\bm{X} \mid \bm{A}, \bm{B}) = \prod_{k \in [T]} p_{\bm{u}}(\bm{B}^{-1} \bm{r}^{(k)})|\det(\bm{B})|^{-1}
        \]
        Therefore, with the substitution $\bm{u}^{(k)} = \bm{B}^{-1}\bm{r}^{(k)}$, the MLE corresponds to minimizing the negative log of the above, divided by $T$:
        \[
            \min_{\bm{A}, \bm{B}, \bm{U}} \log|\det(\bm{B})| + \frac{1}{T\mu}\|\bm{U}\|_1 \text{ s.t. } \bm{X}_+ = \bm{A}\bm{X}^\top + \bm{B}\bm{U}^\top
        \]
        where we have written $p_{\bm{u}}(\bm{u}^{(k)}) \propto \exp(-\mu^{-1}\|\bm{u}^{(k)}\|_1)$. Denote $\bm{R} = \bm{X}_+ - \bm{A}\bm{X}^\top$. Under the change of variables $\bm{B}^{-\top} = \bm{\Phi}$, we have that $\bm{U} = \bm{R}^\top\bm{\Phi}$. At any stationary point of the resulting unconstrained problem in $\bm{\Phi}, \bm{A}$, it follows directly from a statement of the KKT conditions that $\|\bm{R}^\top \bm{\Phi}_i\|_1 = \|\bm{U}_i\|_1 = T\mu$. Therefore, any global optimum is feasible for the problem \eqref{opt:L1_constr}. If this is not globally optimal for \eqref{opt:L1_constr}, there exists another feasible point which achieves lower volume in $\bm{B}$ while remaining feasible, hence contradicting global optimality for the MLE. Suppose there was a global optimum of $(P)$ that is not globally optimal for the MLE, then either the constraint is not tight for some $\bm{U}_i$, which is a contradiction as in that case $(\bm{B}_i, \bm{U}_i) \to (\lambda^{-1} \bm{B}_i, \lambda\bm{U}_i)$ shrinks the objective while maintaining feasibility; or $\log|\det(\bm{B}_P)| + \frac{1}{T\mu}\|\bm{U}_{P}\|_1 > \log|\det(\bm{B}_{MLE})| + \frac{1}{T\mu}\|\bm{U}_{MLE}\|_1$ which implies $\log|\det(\bm{B}_P)| > \log|\det(\bm{B}_{MLE})|$, a contradiction. Therefore, \eqref{opt:L1_constr} is the MLE.
    \end{proof}

\begin{lemma}[Base Norm equivalence]\label{lemma:base_equiv}
    The pair $(\bm{x}, \bm{u})$ is persistently scattering if and only if there exists $\varepsilon > 0$ such that both:
    \begin{enumerate}
        \item $ \|\tilde{\bm{X}}\bm{w} + \tilde{\bm{U}}\bm{v}\|_1^2 \ge \varepsilon\|\bm{w}\|_2^2 + \|\bm{v}\|_2^2$
        \item $\|\tilde{\bm{X}}\bm{w} + \tilde{\bm{U}}\bm{v}\|_1^2 = \varepsilon\|\bm{w}\|_2^2 + \|\bm{v}\|_2^2 \iff \bm{w}=0, \|\bm{v}\|_0 \le 1$
    \end{enumerate}
\end{lemma}
\begin{proof}
    Suppose the pair $(\bm{x}, \bm{u})$ is persistently scattering. Since the columns of $\tilde{\bm{U}}$ have been scaled to unit L1 norm, it follows from H\"older's inqueality that $\{(0, \pm \bm{e}^{(i)})\}_{i \in [m]} \subset \partial \mathcal{Z}$. Therefore, the diagonal matrix in the definition of persistent scattering must satisfy $\bm{D} = \begin{bmatrix}
        \bm{D}_X & 0 \\ 0 & \bm{I}
    \end{bmatrix}$ for $\bm{D}_X$ nonsingular. The support function of this scaled ball satisfies $\sigma_{\bm{D}\mathcal{B}_2}(\bm{w}, \bm{v}) = \sup_{(\bm{a}, \bm{b}) \in \mathcal{B}_2}\langle \bm{D}\bm{a}, \bm{w} \rangle + \langle\bm{b}, \bm{v} \rangle = \sqrt{\|\bm{D}\bm{w}\|_2^2 + \|\bm{v}\|_2^2}$.
    The support function $\sigma_{\mathcal{Z}}(\bm{w}, \bm{v})$ satisfies:
    \begin{align*}
        \sigma_{\mathcal{Z}}(\bm{w}, \bm{v}) 
        &= \sup_{(\bm{a}, \bm{b}) \in \mathcal{Z}} \langle \bm{w}, \bm{a} \rangle + \langle\bm{v}, \bm{b} \rangle \\
        &= \sup_{\bm{z} \in \mathcal{B}_\infty}\langle \tilde{\bm{X}}^\top\bm{z}, \bm{w} \rangle + \langle \tilde{\bm{U}}^\top \bm{z}, \bm{v} \rangle \\
        &= \|\tilde{\bm{X}} \bm{w} + \tilde{\bm{U}}\bm{v} \|_1
    \end{align*}
    Since the support function is the gauge of the polar body, $\bm{D}\mathcal{B}_2 \subseteq \mathcal{Z} \iff \sigma_{\mathcal{Z}}(\bm{w}, \bm{v}) \ge \sigma_{\bm{D}\mathcal{B}_2}(\bm{w}, \bm{v}) \iff \|\tilde{\bm{X}}\bm{w} + \tilde{\bm{U}}\bm{v}\|_1^2 \ge \|\bm{D}_X\bm{w}\|_2^2 + \|\bm{v}\|_2^2$.
    We proceed to prove $\|\tilde{\bm{X}} \bm{w} + \tilde{\bm{U}}\bm{v}\|_1^2 = \|\bm{D}_X \bm{w}\|_2^2 + \|\bm{v}\|_2^2 \iff \bm{w} =0, \|\bm{v}\|_0 \le 1$ is equivalent to $\bm{D} \mathcal{B}_2 \cap \partial \mathcal{Z} = \{(0, \pm \bm{e}^{(i)})\}_{i \in [m]}$ as follows:
    \begin{itemize}
        \item $(\Leftarrow)$ Suppose $\sigma_{\mathcal{Z}}(\bm{y}) = \sigma_{\bm{D}\mathcal{B}_2}(\bm{y})$ for some $\bm{y} \ne 0$. Since $\bm{D}\mathcal{B}_2$ is strictly convex, there is a unique $\bm{p}_* \in \bm{D}\mathcal{B}_2$ satisfying $\langle \bm{p}_*, \bm{y} \rangle = \sigma_{\bm{D}\mathcal{B}_2}(\bm{y})$, which is $\bm{p}_* = \bm{D}^2\bm{y}/\|\bm{D}\bm{y}\|_2$. By hypothesis, $\sigma_{\mathcal{Z}}(\bm{y}) = \langle \bm{p}_*, \bm{y} \rangle$, so $\bm{p}_* \in \partial \mathcal{Z}$, and therefore assuming the boundary intersection condition implies $\bm{p}_* \in \{(0, \pm \bm{e}^{(i)})\}_{i \in [m]}$. Writing $\bm{y} = (\bm{w}, \bm{v})$, it follows by nonsingularity of $\bm{D}_X$ that $\bm{y} = \|\bm{D}\bm{y}\|_2\bm{D}^{-2}\bm{p}_* \implies \bm{w} = 0, \|\bm{v}\|_0 \le 1$. The reverse implication is furnished by homogeneity and $\ell_1$-normalization of the columns of $\tilde{\bm{U}}$.
        \item $(\Rightarrow)$ We show first that $\bm{D}\mathcal{B}_2 \cap \partial \mathcal{Z} \subseteq \{(0, \pm \bm{e}^{(i)}) \}$. Take $\bm{p} \in \bm{D} \mathcal{B}_2 \cap \partial \mathcal{Z}$. Since $\bm{p} \in \partial \mathcal{Z}$, $\exists \bm{y} \ne 0$ such that $\langle \bm{p}, \bm{y} \rangle = \sigma_{\mathcal{Z}}(\bm{y})$. Since $\bm{p} \in \bm{D}\mathcal{B}_2$, we have $\langle \bm{p}, \bm{y} \rangle \le \sigma_{\bm{D}\mathcal{B}_2}(\bm{y}) = \|\bm{D}\bm{y}\|_2 \le \sigma_{\mathcal{Z}}(\bm{y})$, where the last inequality follows by assuming $\bm{D}\mathcal{B}_2 \subseteq \mathcal{Z}$. Therefore all inequalities are tight, and $\sigma_{\mathcal{Z}}(\bm{y}) = \|\bm{D}\bm{y}\|_2$. By hypothesis, $\bm{y} = (0, \bm{v})$ with $\|\bm{v}\|_0 \le 1$. Since $\bm{D}\mathcal{B}_2$ is strictly convex, it follows that $\bm{p} = \bm{D}^2\bm{y}/\|\bm{D}\bm{y}\|_2$ is the unique maximizer in the definition of $\sigma_{\bm{D}\mathcal{Z}}(\bm{y})$, and hence $\bm{p} = (\bm{D}_X^2 0, \bm{v})/\|\bm{v}\|_2 = (0, \bm{v}/\|\bm{v}\|_2)$. Therefore the second component of $\bm{p}$ is unit norm and $1$-sparse, and so is a signed basis vector. 

        We now show $\{(0, \pm\bm{e}^{(i)})\}_{i \in [m]} \subseteq \bm{D}\mathcal{B}_2 \cap \partial \mathcal{Z}$. It is immediate that $(0, \pm \bm{e}^{(i)}) \subseteq \bm{D}\mathcal{B}_2$, since $\|\bm{D}^{-1}(0, \pm \bm{e}^{(i)})\|_2 = \|\bm{e}^{(i)}\|_2 = 1$. Furthermore, each $(0, \pm \bm{e}^{(i)}) \in \partial \mathcal{Z}$, since taking $\bm{y} = (0, \pm \bm{e}^{(i)})$ gives $\langle \bm{y}, \bm{y} \rangle = 1 = \|\tilde{\bm{U}}(\pm\bm{e}^{(i)})\|_1 = \sigma_\mathcal{Z}(\bm{y})$, so $\bm{y}$ achieves the support function in its own direction.
    \end{itemize}
    The persistent scattering condition is thus equivalent to the existence of $\bm{D}_X$ nonsingular such that
    \begin{enumerate}
        \item  $\|\tilde{\bm{X}} \bm{w} + \tilde{\bm{U}}\bm{v}\|_1^2 \ge \|\bm{D}_X \bm{w}\|_2^2 + \|\bm{v}\|_2^2$
        \item $\|\tilde{\bm{X}} \bm{w} + \tilde{\bm{U}}\bm{v}\|_1^2 = \|\bm{D}_X \bm{w}\|_2^2 + \|\bm{v}\|_2^2 \iff \bm{w} =0, \|\bm{v}\|_0 \le 1$
    \end{enumerate}
    The first condition is the support function characterization of $\bm{D} \mathcal{B}_2 \subseteq \mathcal{Z}$ and follows from definition. 
    Suppose $\bm{D}_X$ is any such matrix. It follows that by letting $\varepsilon = \sigma_{\min}(\bm{D}_X)^2$, these conditions imply that $\|\tilde{\bm{X}} \bm{w} + \tilde{\bm{U}}\bm{v}\|_1^2 \ge \varepsilon \|\bm{w}\|_2^2 + \|\bm{v}\|_2^2$, with equality implying $\bm{w} = 0$ and $\|\bm{v}\|_0 \le 1$. Taking $\bm{w} = 0$ and $\|\bm{v}\|_0 \le 1$ yields $\varepsilon\|\bm{w}\|_2^2 + \|\bm{v}\|_2^2 = \|\bm{D}_X\bm{w}\|_2^2 + \|\bm{v}\|_2^2 = \|\tilde{\bm{X}}\bm{w} + \tilde{\bm{U}}\bm{v}\|_1^2$ by the above. The equivalence follows by setting $\bm{D}_X = \sqrt{\varepsilon} \bm{I}$.
\end{proof}

\begin{lemma}[Lifting]\label{lemma:lifting}
    Suppose $X, U$ have columns of unit $\ell_1$ norm. Denote $\Omega(v) := \min_w \|Uv + Xw\|_1$. If
    \begin{enumerate}
        \item $\Omega(v) \ge \|v\|_2$ \label{c1}
        \item $\Omega(v) = \|v\|_2 \iff \|v\|_0 \le 1$ \label{c2}
        \item $\forall i \in [m],  (\Omega(e_i)=\|U_i + X w\|_1 \implies w = 0)$ \label{c3}
    \end{enumerate}
    Then there exists $\varepsilon > 0$ such that $\forall v, w,$ $\|Xw + Uv\|_1^2 \ge \|v\|_2^2 + \varepsilon \|w\|_2^2$
\end{lemma}

\begin{proof}
    Define $g(w, v) := \|Uv + Xw\|_1^2 - \|v\|_2^2$. We have that $$g(w, v) = (\|Uv + Xw\|_1^2 - \Omega(v)^2) + (\Omega(v)^2 - \|v\|_2^2),$$ where the first term is non-negative by definition of $\Omega$, and the second is non-negative by condition \ref{c1}. We first show that $g(w, v) > 0$ for all $w \ne 0$. If $g(w, v) = 0$, then $\|Uv + Xw\|_1 = \|v\|_2$. Since $\Omega(v) \le \|Uv + Xw\|_1$, condition 1 implies $\|v\|_2 \le \Omega(v) \le \|v\|_2$ so $\Omega(v) = \|v\|_2$. Condition \ref{c2} then implies $\|v\|_0 \le 1$, so there exists $t, j$ such that $v = t e^{(j)}$. Therefore, $\|U t e^{(j)} + Xw\|_1 = \|v\|_2 = \Omega(te_j)$, so $w = 0$ by condition 3. Thus, $w \ne 0 \implies g(w, v) > 0$.

    We will next show that the function $h(w) = \inf_v g(w, v)$ is bounded away from zero on the unit sphere. To do so, we will show the existence of $M$ such that the infimum is achieved for $\|v\|_2 \le M$. So let $w^{(i)}$ be a convergent sequence of points on the sphere, $v_*^{(i)}$ any approximate minimizers satisfying $g(w^{(i)}, v_*^{(i)}) \le h(w^{(i)}) + 1/i$, and $\hat{v}_*^{(i)} := v_*^{(i)}/\|v_*^{(i)}\|_2$. Suppose $\|v_*^{(i_\ell)}\|_2 \to \infty$ for any subsequence $\ell \mapsto i_\ell$:
    
    If $\hat{v}_*^{(i_\ell)} \to \hat{v} \ne \pm e^{(j)}$, then by the 2-homogeneity of $g$, $g(w^{(i_\ell)}, v_*^{(i_\ell)}) = \|v_*^{(i_\ell)}\|_2^2  g(w^{(i_\ell)}/\|v_*^{(i_\ell)}\|_2, \hat{v}_*^{(i_\ell)}) \ge \|v_*^{(i_\ell)}\|_2^2(\Omega(\hat{v}_*^{(i_\ell)})^2-1) \to \infty$ where have used condition 2 to enforce that for all sufficiently large $\ell$, $\Omega(\hat{v}_*^{(i_\ell)}) - \|\hat{v}_*^{(i_\ell)}\|_2 \ge c > 0$, contradicting $g(w^{(i_\ell)}, v_*^{(i_\ell)}) \le h(w^{(i_\ell)}) + 1/i \le g(w^{(i_\ell)}, 0) + 1\le C+1$. 
    
    On the other hand, suppose $\hat{v}_*^{(i_\ell)} \to \pm e^{(j)}$. By way of contradiction, further suppose $g(w^{(i_\ell)}, v_*^{(i_\ell)}) \le C+1$ for all $\ell$. Set $t_\ell = \|v_*^{(i_\ell)}\|_2$. Set $s_\ell = w^{(i_\ell)}/t_\ell$. By 2-homogeneity, 
    \[
        g(w^{(i_\ell)}/t_{\ell}, v_*^{(i_\ell)}/t_\ell) =\|U\hat{v}_*^{(i_\ell)} + Xs_\ell\|_1^2 - 1 \le (C+1)/t_\ell^2
    \]
    so $\|U \hat{v}_*^{(i_\ell)} + X s_\ell \|_1 \le 1 + O(1/t_\ell^2)$. By condition (1) and (2), in a neighborhood of $e^{(j)}$ on the sphere, $\Omega(v)$ is uniquely minimized at $e^{(j)}$. Therefore, by piecewise linearity of $\Omega(v)$ (being the infimal projection of a piecewise linear function) and smoothness of the sphere, there exists a neighborhood of $e^{(j)}$ and a constant $c > 0$ such that $\Omega(v) \ge 1 + c\|e^{(j)} - v\|_2$ in this neighborhood. Therefore, for all sufficiently large $\ell$, $c\|e^{(j)} - \hat{v}_*^{(i_\ell)}\|_2 \le \Omega(\hat{v}_*^{(i_\ell)}) - 1 \le O(1/t_\ell^2)$. Now, using directly the piecewise linearity of $\|U v + X s\|_1$ and condition 3, there exists $\varepsilon > 0$ indepedent of $t_\ell$ such that $\|Ue^{(j)} + X s_\ell\|_1 \ge 1 + \varepsilon \|s_\ell\|_2 = 1 + \varepsilon/t_\ell$. By the triangle inequality,
    \begin{align*}
        \|U\hat{v}_*^{(i_\ell)} + Xs_\ell \|_1 &\ge \|U e^{(j)} + X s_\ell \|_1 - \|U\|_{1 \to 1}\|\hat{v}_*^{(i_\ell)} - e^{(j)}\|_1 \\
        &\ge 1 + \varepsilon/t_\ell -O(1/t_\ell^2)
    \end{align*}
    But this must be $\le 1 + O(1/t_\ell^2)$, so $\varepsilon/t_\ell \le O(1/t_\ell^2)$, giving $\varepsilon \le O(1/t_\ell) \to 0$ as $\ell \to \infty$, a contradiction.

    Therefore, there does not exist any subsequence such that $\|v_*^{(i_\ell)}\|_2 \to \infty$. It follows that the set of approximate minimizers is uniformly bounded by some $M$ satisfying $\|v_*(w)\|_2 \le M$ for any $\|w\|_2 =1$, and hence $\inf_v g(w, v) = \inf_{\|v\|_2 \le M} g(w, v)$ is obtained by compactness.
    
    Therefore, $h(w) = \inf_v g(w, v) = g(w, v_*(w)) > 0$. By compactness again, we have that

    \[
        \min_{\|w\|_2 = 1} h(w) =: \varepsilon > 0
    \]
    By the 2-homogeneity of $g$ in $(w, v)$, we may conclude that
    \[
        g(w, v) = \|w\|_2^2 \cdot g\left(\frac{w}{\|w\|_2}, \frac{v}{\|w\|_2}\right) \ge \|w\|_2^2 \cdot h\left( \frac{w}{\|w\|_2} \right) \ge \varepsilon \|w\|_2^2
    \]
    which proves that $\|Uv + Xw\|_1^2 \ge \|v\|_2^2 + \varepsilon\|w\|_2^2$, as desired.
\end{proof}

\begin{proof}(Proposition \ref{prop:support_equiv})\label{proof:support_equiv}
    By lemma \ref{lemma:base_equiv}, the persistent scattering condition is equivalent to the existence of $\varepsilon > 0$ such that
    \begin{enumerate}
        \item $\|\tilde{\bm{X}}\bm{w} + \tilde{\bm{U}}\bm{v}\|_1^2 \ge \varepsilon\|\bm{w}\|_2^2 + \|\bm{v}\|_2^2$
        \item $\|\tilde{\bm{X}}\bm{w} + \tilde{\bm{U}}\bm{v}\|_1^2 = \varepsilon\|\bm{w}\|_2^2 + \|\bm{v}\|_2^2 \iff \bm{w}=0, \|\bm{v}\|_0 \le 1$
    \end{enumerate}
    We will now prove equivalence of this characterization to the proposed statement by reciprocal implication.

    $(\Rightarrow)$ The first term implies that $\Omega_1(\bm{v}) \ge \|\bm{v}\|_2$ by minimization over $\bm{w}$. If $\Omega_1(\bm{v}) = \|\bm{v}\|_2$, then let $\bm{w}^*$ be such that $\|\tilde{\bm{X}} \bm{w}^* + \tilde{\bm{U}}\bm{v}\|_1 = \Omega_1(\bm{v})$. Then the first inequality yields \[
        \varepsilon\|\bm{w}^*\|_2^2 + \|\bm{v}\|_2^2 \le \|\tilde{\bm{U}}\bm{v} + \tilde{\bm{X}} \bm{w}^*\|_1^2 = \|\bm{v}\|_2^2
    \]
    It follows that $\bm{w}^* =0$. Since equality is achieved, the second condition implies $\|\bm{v}\|_0 \le 1$. Conversely, if $\|\bm{v}\|_0 \le 1$, we have that $\bm{v} = \alpha \bm{e}^{(j)}$ for some $j$. Therefore, taking $\bm{w} = 0$, by the definition of $\Omega_1$:
    \[
        \Omega_1(\bm{v}) \le \|\tilde{\bm{U}}\bm{v}\|_1 = \alpha = \|\bm{v}\|_2
    \]
    By the first inequality, we have that $\Omega_1(\bm{v}) = \|\bm{v}\|_2$. Lastly, fix $i \in [m]$. Suppose that $\|\tilde{\bm{X}} \bm{w} + \tilde{\bm{U}}\bm{v}\|_1 = \Omega_1(\bm{e}_i)$. Applying equality (1) at $\bm{v} = \bm{e}_i$ yields $1 = \Omega_1(\bm{e}_i)^2 = \|\tilde{\bm{X}}\bm{w} + \tilde{\bm{U}}_i\|_1^2 \ge \varepsilon\|\bm{w}\|_2^2 + 1$ and therefore, $\bm{w} = 0$.

    $(\Leftarrow)$ Now suppose all three conditions hold. By lemma \ref{lemma:lifting}, there exists a largest $\varepsilon_0 > 0$ such that $\forall v, \forall w, \|Xw + U v\|_1^2 \ge \|v\|_2^2 + \varepsilon\|w\|_2^2$. Now let $\varepsilon < \varepsilon_0$, then $\|Uv + Xw\|_1^2 \ge \varepsilon \|w\|_2^2 + \|v\|_2^2$. Suppose equality is obtained. Since we have selected $\varepsilon  < \varepsilon_0$, equality is impossible for any $w \ne 0$. So under equality, we must have that $\|X w + Uv\|_1^2 = \|Uv\|_1^2 = \|v\|_2^2$, so by conditions 1 and 2, $\|v\|_0 \le 1$. Conversely, suppose $\|v\|_0 \le 1$ and $w = 0$, then $v = te_j$ and so $\|Uv\|_1 = t=\|v\|_2$, and so since $w = 0$, $\|Uv + Xw\|_1^2 = \varepsilon\|w\|_2^2 + \|v\|_2^2$. 
\end{proof}

   \begin{lemma}\label{lemma:equiv_cond}
        Suppose that $\bm{X}_+ = \bm{A}_\natural\bm{X}^\top + \bm{B}_\natural\bm{U}_\natural^\top$ for $(\bm{x}, \bm{u}_\natural)$ persistently exciting and $\ker(\bm{B}_\natural) =\{0\}$. Denote $\bm{P} = \bm{I} - \bm{X}\bm{X}^\dagger$. Then \[
            \bm{X}_+ = \bm{A}\bm{X}^\top + \bm{B}\bm{U}^\top \iff \exists \bm{\Phi},
            \begin{cases}
                \bm{A} = (\bm{X}_+ - \bm{B}\bm{U}^\top)(\bm{X}^\top)^\dagger \\ 
                \bm{P}\bm{U} = \bm{P}\bm{U}_\natural\bm{\Phi} \\ \bm{B}_\natural = \bm{B}\bm{\Phi}^\top
            \end{cases}
        \]
    \end{lemma}
    \begin{proof}(Lemma \ref{lemma:equiv_cond})
        Under the assumptions, $\bm{X}_+\bm{P} = \bm{B}_\natural \bm{U}_\natural^\top \bm{P}$, $\bm{X}$ is full column rank, and $\bm{U}_\natural^\top \bm{P}$ is full row rank (by persistency of excitation). Therefore, $\bm{X}_+\bm{P}$ is rank $m$.
        \begin{align*}
            &(\bm{X}_+ = \bm{A}\bm{X}^\top + \bm{B}\bm{U}^\top) \\
            &\iff (\bm{X}_+\bm{P} = \bm{B}\bm{U}^\top\bm{P}) \wedge (\bm{X}_+(\bm{I} - \bm{P}) = \bm{A}\bm{X}^\top + \bm{B}\bm{U}^\top(\bm{I} - \bm{P})) \\
            &\iff (\bm{B}_\natural\bm{U}_\natural^\top\bm{P} = \bm{B}\bm{U}^\top\bm{P}) \wedge ((\bm{X}_+ - \bm{B}\bm{U}^\top)(\bm{I} - \bm{P}) = \bm{A}\bm{X}^\top )\\
            &\iff (\bm{B}_\natural\bm{U}_\natural^\top\bm{P} = \bm{B}\bm{U}^\top\bm{P}) \wedge ((\bm{X}_+ - \bm{B}\bm{U}^\top)(\bm{X}^\top)^\dagger = \bm{A})
        \end{align*}
        Since $\bm{X}_+\bm{P} = \bm{B}_\natural\bm{U}_\natural^\top \bm{P} = \bm{B}\bm{U}^\top\bm{P}$, any feasible $\bm{B}$ must have full column rank. Therefore, we can write $\bm{B}_\natural = \bm{B}\bm{\Phi}^\top$ for some invertible $\bm{\Phi}$. So under the assumptions, we have
        \[
            \bm{B}_\natural\bm{U}_\natural^\top \bm{P} = \bm{B}\bm{U}^\top \bm{P} \iff (\bm{P}\bm{U}_\natural \bm{\Phi} = \bm{P}\bm{U})\wedge (\bm{B}_\natural = \bm{B}\bm{\Phi}^\top).
        \]
    \end{proof}


    \begin{proof}(Proposition \ref{prop:svd_opt})\label{proof:svd_opt}
        By lemma \ref{lemma:equiv_cond}, every feasible $\bm{A}, \bm{B}, \bm{U}$ for \eqref{opt:L1_constr} satisfies $\bm{X}_+\bm{P}_X^\perp = \bm{B}\bm{U}^\top \bm{P}_{X}^\perp$, with $\bm{A}$ uniquely determined as $\bm{A} = (\bm{X}_+ - \bm{B}\bm{U}^\top)(\bm{X}^\top)^\dagger$, and $\bm{B}$ necessarily full column rank. So \eqref{opt:L1_constr} is equivalent to
        \[
            \min_{\bm{B}, \bm{U}} \vol(\bm{B}) \text{ s.t. } \bm{B}(\bm{P}_X^\perp\bm{U})^\top = \bm{X}_+ \bm{P}_X^\perp, \, \max_{i}\|\bm{U}_i\|_1 \le \mu T
        \]
        Since $\bm{X}_+ \bm{P}_X^\perp = \bm{Q}\bm{\Sigma}\bm{V}^\top$ has rank $m$ and $\bm{B}$ is full column rank, the row space of $\bm{B}(\bm{P}_X^\perp \bm{U})^\top$ is equal to the row space of $(\bm{P}_X^\perp \bm{U})^\top$, which must therefore equal $\text{img}(\bm{V})$. Hence $\bm{P}_X^\perp\bm{U} = \bm{V}\bm{\Phi}$ for a unique invertible $\bm{\Phi}$, and the constraint becomes $\bm{B}\bm{\Phi}^\top \bm{V}^\top = \bm{Q}\bm{\Sigma}\bm{V}^\top \iff \bm{B}\bm{\Phi}^\top = \bm{Q}\bm{\Sigma} \iff \bm{B} = \bm{Q}\bm{\Sigma}\bm{\Phi}^{-\top}$. Since $\bm{Q}$ has orthonormal columns, 
        \[
            \bm{B}^\top\bm{B} = \bm{\Phi}^{-1}\bm{\Sigma}^2\bm{\Phi}^{-\top} \implies \vol(\bm{B}) = \sqrt{\det(\bm{B}^\top\bm{B})} = \frac{\det(\bm{\Sigma})}{|\det(\bm{\Phi})|}
        \]
        Since $\det(\bm{\Sigma})$ is a fixed constant, minimizing $\vol(\bm{B})$ is equivalent to minimizing $-\log|\det(\bm{\Phi})|$. Putting things together, we have established that \eqref{opt:L1_constr} is equivalent to 
        \[
            \min_{\bm{\Phi}, \bm{U}} -\log |\det(\bm{\Phi})| \text{ s.t. } \begin{cases}
                \bm{V}\bm{\Phi} = \bm{P}_X^\perp \bm{U} \\
                \max_{i} \|\bm{U}_i\|_1 \le \mu T
            \end{cases}
        \]
        where the global optima $(\bm{\Phi}_\star, \bm{U}_\star)$ are in bijective correspondence as indicated in the statement.
        
    \end{proof}

\end{document}